\documentclass[letter,10pt]{article}
\usepackage[utf8]{inputenc}

\usepackage{paithan}  

\newcommand{\fuzzy}{\ensuremath{\mathcal{N}}}
\newcommand{\zero}{\ensuremath{\mathcal{P}}}

\title{Games from Basic Data Structures}
\author{Mara Bovee \\University of New England\\bovee.mara@gmail.com \and Kyle Burke\\Plymouth State University\\ kgburke@plymouth.edu \and Craig Tennenhouse\\University of New England\\ctennenhouse@une.edu}

\begin{document}

\maketitle

\begin{abstract}
  In this paper, we consider combinatorial game rulesets based on data structures normally covered in an undergraduate Computer Science Data Structures course: arrays, stacks, queues, priority queues, sets, linked lists, and binary trees.  We describe many rulesets as well as computational and mathematical properties about them.  Two of the rulesets, \ruleset{Tower Nim} and \ruleset{Myopic Col}, are new.  We show polynomial-time solutions to \ruleset{Tower Nim} and to \ruleset{Myopic Col} on paths.
\end{abstract}

\section{Introduction}

\subsection{Notation}

\begin{itemize}
    \item For an integer $x$ we will use $x^-$ to denote an integer less than or equal to $x$, and $x^+$ an integer at least as large as $x$.

    \item We use the $\oplus$ operator to refer to the bitwise binary XOR operation.  For example, $3 \oplus 5 = 6$.

    \item Other notation common to combinatorial game theory is mentioned below.
\end{itemize}

\subsection{Combinatorial Game Theory}

Two-player games with alternating turns, perfect information, and no random elements are known as \emph{combinatorial games}.  Combinatorial Game Theory concerns determining which of the players has a winning move from any \emph{position} (game state).  Some important considerations here follow.

\begin{itemize}
    \item  This paper only considers games under \emph{normal play} conditions, meaning that the last player to make a move wins.  (If a player can't make a legal move on their turn, they lose.)
    \item  The sum of two game positions is also a game where a legal move consists of moving on one of the two position summands.  The value of a game sum is simply the sum of the values of the summands.
    \item  Many of the sections require the reader to be moderately comfortable with number values, nimbers (or Grundy values), and outcome classes (specifically, $\fuzzy$ and $\zero$).
    \item Understanding of the game-set notation (e.g. $\{*, 1 | *, 2\}$) is necessary for understanding Lemma \ref{lem:numbersPlusStar}.
\end{itemize}

The authors recommend \cite{LessonsInPlay:2007} and \cite{WinningWays:2001} for background texts on combinatorial game theory.

\subsection{Algorithmic Combinatorial Game Theory}

Algorithmic combinatorial game theory concerns the computational difficulty in determining the winnability, outcome class, or value of game positions.  Given a position, how long does it take for an algorithm to determine whether the next player has a winning strategy?  For many rulesets, this can be solved in time polynomial in the size of the representation of the position.  These ``easy'' rulesets usually allow some sort of closed formula to determine the winner.

For others, the problem is known to be intractable.  For example, solving \ruleset{Snort} is known to be \cclass{PSPACE}-complete\cite{DBLP:journals/jcss/Schaefer78}, meaning that unless an efficient algorithm is found to solve all other problems in \cclass{PSPACE}, no polynomial-time algorithm can determine the winnability of all \ruleset{Snort} positions.  For these ``hard'' games, usually the majority of the induced game tree must be inspected in order to determine the winner.

In algorithmic combinatorial game theory, a ruleset is considered \emph{solved} if either the polynomial time ``trick'' is known, or if it has been shown to be intractable.  Many rulesets remain unsolved: it is unknown whether they can be evaluated efficiently or not.

In this paper, we discuss known algorithmic results about each ruleset.

\subsection{Data Structures}

Data structures are logical structures used to store digital data in computing systems.  Different structures facilitate different operations.  For example, a max-heap can quickly add and remove elements, but is restricted to removing only the greatest element.

A data structures course in computer science is a common course in the basic programming sequence in many undergraduate curricula.  The topics cover many logical structures to store digital data.  This paper covers games based on: arrays, sets, stacks, queues, priority queues, linked lists, and binary trees.  We define each of these in the following sections.

\subsection{Organization}

The following 7 sections are each devoted to a ruleset based on a particular data structure.  We provide Table \ref{table:organization} as a handy key to corresponding data structures, ruleset names, and section numbers.

\begin{table}[h!]
	\begin{center}\begin{tabular}{|l|l|l|}
		\hline
		Data Structure & Ruleset & Section\\ \hline
		Array & \ruleset{Nim} & \ref{section:arrays} \\ \hline
		Sets & \ruleset{Antonim} & \ref{section:sets} \\ \hline		
		Stacks & \ruleset{Tower Nim} & \ref{section:stacks}\\ \hline
		Queues & \ruleset{Rotisserie Nim} & \ref{section:queues}\\ \hline
		Priority Queues & \ruleset{Greedy Nim} & \ref{section:priorityQueues} \\ \hline
		Linked Lists & \ruleset{Myopic Col} & \ref{section:linkedLists} \\ \hline
		Binary Trees & \ruleset{Myopic Col} & \ref{section:binary-trees} \\ \hline
		Graphs & many games & \ref{section:graphs} \\ \hline
	\end{tabular}\end{center}
	\caption{Corresponding data structures, rulesets, and sections.}\label{table:organization}
\end{table}

\section{Arrays}
\label{section:arrays}

An array is a linear data structure addressable by the natural numbers.  Each array element is stored in a cell so that the first $n$ elements have respective addresses $0, 1, \ldots, n-1$.  

\subsection{Ruleset: Nim}

\ruleset{Nim} is a simple combinatorial game that has been well-studied\cite{Bouton:1901}.  
\begin{definition}[Nim]
	\label{def:nim}
	\ruleset{Nim} is an impartial combinatorial game where each position is represented by a list of natural numbers: $K = (k_0, k_1, \ldots, k_{n-1})$, $k_i \in \mathbb{N}$.  The position represented by $M = (m_0, m_1, \ldots, m_{n-1})$ is an option of $K$ if $\exists i \in 0, 1, \ldots, n-1: $
	\begin{itemize}
		\item $\forall j \neq i: k_j = m_j$, and
		\item $m_i < k_i$
	\end{itemize} 
\end{definition}

The grundy value of any \ruleset{Nim} position can be found in linear time: $\mathcal{G}(k_0, k_1, \ldots, k_{n-1})= k_0 \oplus k_1 \oplus \cdots \oplus k_{n-1}$.  There is a winning move from position $K$ exactly when the grundy value, $\mathcal{G}(K)$, is not zero.\cite{Bouton:1901}

\section{Sets}
\label{section:sets}

A set is a data structure where there is no implicit order on the elements, and copies of elements are not recorded, (each element appears at most once).  

\subsection{Ruleset: Antonim}

Instead of keeping a list (or bag) of heaps as in \ruleset{Nim}, we employ a set.  In the resulting \ruleset{Antonim}, each turn consists of choosing a heap from the set, removing some number of sticks from that heap, then replacing the original heap with the new heap.  (If the resulting heap has no sticks, we don't bother to add it back.)  The catch here is that each heap is simply a natural number, so if there is already a heap of that size in the set, the new heap is forgotten about; that player might as well have taken all the sticks in that heap.  We provide a formal description in Definition \ref{def:antonim}.

\begin{definition}[Antonim]
	\label{def:antonim}
	\ruleset{Antonim} is an impartial combinatorial game where each position is represented by $S \subset \mathbb{N} \setminus \{0\}$.  The position represented by $S'$ is an option of $S$ if:
	\begin{itemize}
		\item $\exists x \notin S'$ such that $S = S' \cup \{x\}$, or
		\item $\exists x \in S, y \in S'$ such that $x \notin S'$ and $y < x$
	\end{itemize} 
\end{definition}

Finding the outcome class for  \ruleset{Antonim} has already been solved for positions with up to three piles \cite{WinningWays:2001}.  These rules are as follows:

\begin{itemize}
    \item $\{\}$ is a terminal position; so the position is in $\mathcal{P}$.
    \item One pile: $\{k\} \in \mathcal{N}$, as there is always a winning move (by taking all sticks).
    \item Two piles: $\{a, b\} \in 
        \begin{cases}
            \mathcal{P}, & \{a, b\} = \{2k+1, 2k+2\} \mbox{ for some } k \in \mathbb{N} \cr
            \mathcal{N}, & \mbox{otherwise}
        \end{cases}$
    \item Three piles: $\{a, b, c\} \in
        \begin{cases}
            \mathcal{P}, & (a+1) \oplus (b+1) \oplus (c+1)=0 \cr
            \mathcal{N}, & \mbox{otherwise}
        \end{cases}$
\end{itemize}

This last case can be proven inductively or by a reduction to standard Nim. A winning move is to add one stone to each heap, make an appropriate move in Nim, then remove a stone from each heap. The appropriate Nim move will not eliminate any heap entirely since that would require two equal-sized heaps to be present, which is not possible in Antonim. Play continues until the game is down to two heaps.





There is currently no known polynomial time algorithm for determine the outcome class for an Antonim position containing four or more heaps, and computer analysis has not yet shed light on a possible method.

\section{Stacks}
\label{section:stacks}

A stack is a data structure where elements are added and removed.  Each remove operation pulls out the most recently-added element.  We can represent this with a list where items are added and removed from the same end, called the top.  The remove operation (called \emph{pop}) then always fetches the most recently added (\emph{pushed}) element.

Consider the following stack (the top is on the right-hand side): $(5, 2, 4)$.  If we pop an element off the stack, $(5, 2)$ will remain.  If we then push $1$ followed by $6$, followed by another, pop, the resulting stack will be $(5, 2, 1)$.

\subsection{Ruleset: Tower Nim}

\ruleset{Tower Nim} is another spin-off of \ruleset{Nim}.  The difference here is that the heaps are arranged linearly.  A move consists of removing sticks from the top (or front) heap.  If all sticks in one heap are removed, then that heap is removed from the stack.  \ruleset{Tower Nim} is similar to \ruleset{End Nim} \cite{DBLP:journals/combinatorics/AlbertN01}
 but with the added restriction that only one end is accessible. We formalize this in Definition \ref{def:towerNim}.

\begin{definition}[\ruleset{Tower Nim}]
	\label{def:towerNim}
	\ruleset{Tower Nim} is an impartial ruleset where positions are described by a list of non-zero natural numbers: $L = (a_0, a_1, \ldots, a_{n-1}, a_n)$.  The position represented by $L'$ is an option of $L$ if either
	\begin{itemize}
		\item  $L' = (a_0, \ldots, a_{n-1})$, or
		\item  $\exists b \in \mathbb{N} \setminus \{0\}$ such that $b < a_n$ and $L' = (a_0, \ldots, a_{n-1}, b)$.
	\end{itemize}
\end{definition}

After a handful of games, it becomes clear that the winnability of the game is often determined by whether the top heap has 1 stick.

\begin{observation}
\label{obs:stackAllOnes}
 	For games with a stack of all ones, $(1, 1, \ldots, 1)$, the nimber is the xor-sum of all heaps.
\end{observation}

\begin{lemma}
\label{lem:towerNimBigTop}
	When only the top is non-one, the position is always in \fuzzy.
\end{lemma}
\begin{proof}
	Let $(1, 1, \ldots, 1, x)$ represent a position, $G$, of \ruleset{Tower Nim}, where $x > 1$.  There are two cases, based on the length of the list.
	\begin{itemize}
		\item If the position has an even number of heaps, then the next player can take $x-1$ sticks, reducing the game to a zero position.
		\item If the position has an odd number of heaps, then the next player can take all $x$ sticks in the top heap, reducing the number of sticks to an even amount, a zero position.
	\end{itemize}
	Since the current player always has a winning move, $G$ must be in \fuzzy.
\end{proof}

\begin{corollary}
\label{corol:towerOnesOnTop}
	Let $(\ldots, x, 1, \ldots, 1)$ represent a position, $G$, of \ruleset{Tower Nim}, where $x > 1$ and $n$ is the number of 1-heaps above $x$.  Then the nimber of $G$ is the parity of $n$.
\end{corollary}
\begin{proof}
	After $n$ turns, the current position must be represented by $(\ldots, x)$, which is in \fuzzy.  
\end{proof}

\begin{lemma}
	Let $(\ldots, 1, x)$ represent a position, $G$, of \ruleset{Tower Nim}.  There are two cases for the nimber of $G$:
	\begin{itemize}
		\item[$x = 1$] Then either $G = 0$ or $G = *$.
		\item[$x \geq 2$] Then $G = *x$.
	\end{itemize}
\end{lemma}
\begin{proof}
	The first case is taken care of for us, using the previous corollary.
	The second case can be proven by strong induction on $x$.  In the base case, $G$ has both $0$ and $*$ and options, by the previous corollary, and must equal $*2$.  For the inductive case, when $x = k+1$, it has all the children that $x = k$ has $(0, *, *2, \ldots, *(k-1))$, as well as $*k$, but does not include $*(k+1)$.  Thus, by the mex rule, $G = *(k+1)$.
\end{proof}

\begin{corollary}
	Let $(\ldots, y, x)$, where $y \geq 2$, represent a position, $G$, of \ruleset{Tower Nim}.  There are two cases for the nimber of $G$:
	\begin{itemize}
		\item[$x = 1$] Then $G = 0 (G \in \zero)$.
		\item[$x \geq 2$] Then $G \in \fuzzy$.
	\end{itemize}
\end{corollary}
\begin{proof}
	The first case occurs because the only option is a non-zero position as described in the previous lemma.  The second case occurs because the first case is an option.
\end{proof}

With these results, we can determine the outcome class of any position by counting the number of consecutive size-1 heaps on top of the stack and whether there is a bigger heap underneath the consecutive 1-heaps. 

\begin{corollary}[\ruleset{Tower Nim} is in \cclass{P}]
    The outcome class of any \ruleset{Tower Nim} position can be determined in \cclass{P}.
\end{corollary}

\begin{proof}
    Consider a \ruleset{Tower Nim} position $G = (x_{n-1}, \ldots, x_1, x_0)$.  Then:
    \begin{itemize}
        \item if $\forall i: x_i = 1$ then, by Observation \ref{obs:stackAllOnes} $G = 
            \begin{cases}
                0\ (G \in \zero), & n \mbox{ is even}\cr
                *\ (G \in \fuzzy), & n \mbox{ is odd}
            \end{cases}$
        \item otherwise, let $x_k$ be the non-one heap closest to the top.  Now, by Corollary \ref{corol:towerOnesOnTop} $G \in
            \begin{cases}
                \fuzzy, & k \mbox{ is even}\cr
                \zero, & k \mbox{ is odd}
            \end{cases}$
    \end{itemize}
\end{proof}

\section{Queues}
\label{section:queues}

Queues are another data structure very similar to stacks.  The only difference is that the remove operation extracts the least-recently added element.  We can represent this as a list:
queues push elements to the \emph{back} of the list and pop from the \emph{front}, the opposite end.

Consider the following queue (the front is the left-hand side, the back is on the right): $(5, 2, 4)$.  If we pop an element from the queue, $(2, 4)$ will remain.  If we then push $1$ followed by $6$, followed by another, pop, the resulting queue will be $(4, 1, 6)$.

\subsection{Ruleset: Rotisserie Nim}

\ruleset{Rotisserie Nim} is the queue-counterpart to \ruleset{Tower Nim}.  In this game, the next player removes sticks from the heap in the front of the queue.  Then, if there are any sticks left in that heap, it's queued back in to the back of the line.  We formalize this in Definition \ref{def:rotisserieNim}. This is equivalent to \ruleset{Adjacent Nim}, which is itself a special case of \ruleset{Vertex Nim} played on a directed cycle\cite{DBLP:journals/tcs/DucheneR14}. 


\begin{definition}[\ruleset{Rotisserie Nim}]
	\label{def:rotisserieNim}
	\ruleset{Rotisserie Nim} is an impartial ruleset where positions are described by a list of non-zero natural numbers: $L = (a_0, a_1, \ldots, a_n)$.  The position represented by $L'$ is an option of $L$ if either
	\begin{itemize}
		\item  $L' = (a_1, \ldots, a_n)$, or
		\item  $\exists b \in \mathbb{N} \setminus \{0\}$ such that $b < a_0$ and $L' = (a_1, \ldots, a_n, b)$.
	\end{itemize}
\end{definition}


We begin by characterizing two heap games.

\begin{theorem}\label{thm:adjTwoHeaps}
If $L=(a_0,a_1)$, then $L\in \mathcal{N}$ iff $a_0>a_1$.
\end{theorem}
\begin{proof}
First, assume $a_0>a_1$. The next player can move to the position $(a_1,a_1)$, and continue to match the other player's moves until reaching the terminal position. Next, assume that $a_1\geq a_0$. The only move from $L$ is to a position of the form $(a_1,a_0^-)$. Since $a_0^-<a_0\leq a_1$, this is an $\mathcal{N}$-position by the above argument, and therefore $L\in \mathcal{P}$.
\end{proof}

The following theorem, proven in \cite{DBLP:journals/tcs/DucheneR14}, completely characterizes the classes $\mathcal{P}$ and $\mathcal{N}$ where each heap has size at least two.

\begin{theorem}\label{thm:adjNim}
Let $L = (a_0, a_1, \ldots, a_n),a_i\geq 2,  \forall 0\leq i\leq n$ be a position in \ruleset{Rotisserie Nim}, and let $a_-=\min\{a_i\}_0^n$. Then $L\in \mathcal{N}$ iff $n$ is odd or the smallest index $j$ for which $a_j=a_-$ is even.
\end{theorem}

We extend this to some small positions with heaps of size one. 

\begin{theorem}\label{thm:adjThreeHeaps1}
If $a_0=1$ then $L\in \mathcal{P}$ if $a_1>a_2$.
\end{theorem}
\begin{proof}
Assume $a_1>a_2$ and $a_0=1$. The only valid move is to $(a_1,a_2)$ which is in $\mathcal{N}$ by Theorem \ref{thm:adjTwoHeaps}. 
\end{proof}

In order to complete the characterization of all three heap positions in \ruleset{Rotisserie Nim} we must first prove the following lemma.

\begin{lemma}\label{lem:adjStrategy}
Let $L=(a_0,a_1,\ldots ,a_n)$ is a position in \ruleset{Rotisserie Nim}. If $(a_1,\ldots ,a_n)\in \mathcal{N}$ then:
\begin{itemize}
	\item If $(n+1)$ is odd and $L\in \mathcal{N}$ then $(a_1,\ldots ,a_n, 1)\in \mathcal{P}$, and 
	\item If $(n+1)$ is even and $L\in \mathcal{N}$ then $(a_1,\ldots ,a_n, a_0-1)\in \mathcal{P}$.
\end{itemize}
\end{lemma}
\begin{proof}
If $L$ has an even number of heaps, then the same player who encounters heap $a_0$ will also encounter what remains of this heap on a subsequent turn. Therefore, if removing the entire heap does not move the game into a $\mathcal{P}$-position then there is no advantage to removing more than a single stick. Similarly, with an odd number of heaps there is no advantage to leaving more than a single stick for the other player to encounter on a subsequent turn.
\end{proof}

Lemma \ref{lem:adjStrategy} provides an interesting strategy, which leads to a win of if $L\in \mathcal{N}$. Using this strategy we prove another lemma necessary to complete our characterization of three heap games. 

\begin{lemma}\label{lem:adjCompare}
If $L=(a_0,a_1,\ldots)\in \mathcal{P}$ then so is $L'=(a_0^-,a_1^+,a_2^-,a_3^+,\ldots)$.
\end{lemma}
\begin{proof}
Assume to the contrary that $L\in \mathcal{P}$ and $L'\in \mathcal{N}$. Therefore there is a move from $L'$ to a $\mathcal{P}$-position $(a_1^+,a_2^-,a_3^+,\ldots ,a_0^{--})$. That means that a similar move can take $L$ to $(a_1,a_2,a_3,\ldots, a_0^{--})$ since $a_0^-\leq a_0$. In fact, each move to a position in $\mathcal{P}$ from the starting position of $L'$ can be copied by the first player in the starting position $L$, and hence $L\in \mathcal{N}$. This contradicts our assumption, hence $L'\in \mathcal{P}$ as well.
\end{proof}

\begin{theorem}\label{thm:adjThreeHeaps2}
If $a_0>1$ then $L\in \mathcal{P}$ iff $a_1>1$ and $a_2=1$.
\end{theorem}
\begin{proof}
Assume $L=(a_0,2,1)$. Since $(2,1)\in \mathcal{N}$ by Theorem \ref{thm:adjTwoHeaps}, Lemma \ref{lem:adjStrategy} tells us that the first move should be to the position $(2,1,1)$. The second player can then move to $(1,1)$, which is in $\mathcal{P}$. Therefore, $L\in \mathcal{P}$. By Lemma \ref{lem:adjCompare}, $(a_0,a_1,1)\in \mathcal{P}$ whenever $a_2\geq 2$.

Next, assume $a_0,a_2>1$ and $a_1>a_2$. Because $(a_1,a_2)\in \mathcal{N}$ by Theorem \ref{thm:adjTwoHeaps}, by Lemma \ref{lem:adjStrategy} the first player will move to position $(a_1,a_2,1)$. Similarly, because $(a_2,a)\in \mathcal{N}$ the next player will move to $(a_2,1,1)$. The first player can then move to $(1,1)\in \mathcal{P}$. Therefore, $(a_0,a_1,a_2)\in \mathcal{N}$ when $a_0,a_2>1$ and $a_1>a_2$, and hence when $a_0>1$, the position $(a_0,a_1,a_2)\in \mathcal{P}$ iff $a_2=1$. 

Finally, consider the case where $a_0>1$ and $a_1<a_2$. Since $(a_1,a_2)\in \mathcal{P}$ by Theorem \ref{thm:adjTwoHeaps}, t$(a_0,a_1,a_2)\in \mathcal{N}$. Therefore, if $a_0>1$ then $(a_0,a_1,a_2)\in \mathcal{P}$ iff $a_1>1$ and $a_2=1$.
\end{proof}

\section{Priority Queues}
\label{section:priorityQueues}

Priority queues are similar to stacks and queues, except that the order of elements does not necessarily determine the order they will be removed.  When a remove-operation is called, the largest element is always removed.

Consider the following priority queue: $(5, 2, 4)$.  If we remove an element from the queue, $(2, 4)$ will remain.  If we then add $1$ followed by $6$, followed by another, remove, the resulting priority queue will be $(2, 4, 1)$.

\subsection{Ruleset: Greedy Nim}

\ruleset{Greedy Nim} is a ruleset just like \ruleset{Nim}, except that players must choose from a heap with the greatest size to remove sticks from\cite{WinningWays:2001}.  We give a formal version in Definition \ref{def:greedyNim}.


\begin{definition}[Greedy Nim]
	\label{def:greedyNim}
	\ruleset{Greedy Nim} is an impartial ruleset where positions are described by a multiset of non-zero natural numbers, $M$.  $M'$ is an option of $M$ exactly when $\exists x \in M: \forall y \in M: x \geq y$ and $\#_M(x) = 1 + \#_{M'}(x)$ and either:
	\begin{itemize}
		\item $\forall a \in (M \cup M')$ either $x = a$ or $\#_M(a) = \#_{M'}(a)$, or
		\item $\exists b \in M \setminus \{x\}: \#_M(b) + 1 = \#_{M'}(b)$ and $\forall a (\neq b) \in (M \cup M')$ either $x = a$ or $\#_M(a) = \#_{M'}(a)$.
	\end{itemize}
\end{definition}

This game has already been solved; a polynomial-time algorithm exists to determine the outcome class of the game\cite{MR2056015}.

\begin{theorem}[Greedy Nim]
	A \ruleset{Greedy Nim} position, $G$, is in \zero exactly when there are an even number of heaps with the greatest size.
\end{theorem}

\section{Linked Lists}
\label{section:linkedLists}

A linked list is a data structure with more explicit structure than a stack, queue, or priority queue.  Instead of describing which element gets removed, a linked list is a method for organizing elements.  Indeed, stacks and queues can both be implemented using linked lists.

Each element in a linked list exists inside a list node, which also includes a reference to the next list node.  (We consider only singly-linked lists, meaning each node does not have a reference to the previous node.)  In other terms, a linked list is like a directed path graph, with a value at each vertex.

\subsection{Ruleset: Myopic Col}\label{sec:myopiccol}\label{sec:myopiccol}\label{sec:myopiccol}\label{sec:myopiccol}\label{sec:myopiccol}

There are many choices of rulesets that could be used on directed paths.  In order to deviate from impartial games, we use a variant of \ruleset{Col}, \ruleset{Myopic Col}.  

In \ruleset{Col}, a turn consists of a player choosing an uncolored vertex and painting it their own color (blue or red).  However, they not allowed to choose a vertex adjacent to another vertex with their color.

\ruleset{Myopic Col} is played on a directed graph instead of undirected.  When a player chooses a vertex, they are only restricted from choosing vertices with arcs pointing \emph{to} neighboring vertices of their color.  Arcs directed towards the chosen vertex do not need to be considered.

As with other games discussed, we provide a rigorous definition.

\begin{definition}[Myopic Col]
	\label{def:myopicCol}
	\ruleset{Myopic Col} is a combinatorial game where positions are described by a directed graph, $G = (V, E)$, and a coloring of vertices either uncolored or painted red or blue, $(c: V \rightarrow \{uncolored, red, blue\})$.  Each arc, $(a, b) \in E$ points from $a$ to $b$.  Another graph-coloring pair $(G', c')$ is an option of $(G, c)$ for player $A \in \{red, blue\}$ exactly when both
	\begin{itemize}
	  \item $G' = G$, and
	  \item $\exists v \in V:$ 
	  \begin{itemize}
	    \item $\forall v' \in V \setminus \{v\}: c(v') = c'(v')$, and
	    \item $c(v) = uncolored$ and $c'(v) = A$ and $\nexists (v,b) \in E: c(b) = A$
	   \end{itemize}
	\end{itemize}
\end{definition}

The positions we consider in this section are on paths. ...

Consider a positions on a paths described as a list of the colors.  For example, $(blue, red, uncolored, uncolored, blue)$ represents path with five vertices with arcs going from-left-to-right.

Note that the game on a path with an arc from a colored vertex to another vertex is equivalent to the sum of two paths created by removing that arc.  Thus, we must only consider situations where there is at most one colored vertex, and that vertex is at the end of the path.

Before we prove things about the Col paths, we need a general lemma about games with numeric and star values:

\begin{lemma}
\label{lem:numbersPlusStar}
  $\forall x \in \mathbb{R}:$
  \begin{align}
    x + * = \{x | x\}\\
    x = \{x + * | x + *\}
  \end{align}
\end{lemma}
\begin{proof}
  First we prove $(1)$: 
  \begin{align*}
    \{ x | x \} & = x \pm 0 & \mbox{(by the definition of switches)} \\
    & = x + \{0|0\}\\
    & = x + *
  \end{align*}
  
  Next we prove $(2)$ by showing that $\{x + * | x + *\} - x = 0$. Notice that either player loses by playing on $\{x + * | x + *\}$, because $x + * - x = * \in \outcomeClass{N}$.  However, since $- x$ is a number and $\{x + * | x + *\}$ is not, it is still better for both players to make that losing move.  Thus, the sum is a loss for both players and equals $0$.
\end{proof}

\begin{lemma}\label{lem:pathMyopicColValues}
  Let $G$ be a Myopic Col game played on a path with $n$ uncolored vertices, possibly followed by a colored vertex, $v$.  Then $G = 
  \begin{cases}
    n \times *, & \mbox{no } v\\
    (n-1) \times * - 1, & v\mbox{ is blue}\\
    (n-1) \times * + 1, & v\mbox{ is red}
  \end{cases}$
  
\end{lemma}

\begin{proof}

  We proceed by strong induction on $n$:
  
  Base Case: $n = 1$.
  
  All three cases hold.  In the first case, there is only one vertex that either can play on, so $G = *$, as the lemma states.  In the second case, there is a single vertex only Right can color.  Thus, $G = -1$.  In the third case, $G = 1$ by the same logic.
  
  Inductive Case: Assume $n > 1$ and both lemmas are true for all $0 < k < n$.
  
  Consider $G$, the game played on a path with $n$ uncolored vertices.  The moves for Left can be derived from the inductive hypotheses.  Starting by coloring the head, these are: 
  \begin{align*}
    (n-1) \times *, \\
    (0 \times * - 1) + (n-2) \times * & = (n-2) \times * - 1, \\
    (1 \times * - 1) + (n-3) \times * & = (n-2) \times * - 1, \\ 
    \vdots\\ 
    ((n-3) \times * - 1) + 1 \times * & = (n-2) \times * - 1, \\
    ((n-2) \times * - 1) + 0 \times * & = (n-2) \times * - 1
  \end{align*}
  Except for the first one, these are all equal to $(n-2) \times * - 1$.  Since $(n-2)\times * - 1 < (n-1) \times *$, they are all dominated by that first move.  Thus, the best move for Left is to color the head of the path.
  
  The same logic holds to show that Right's best move is also to color the head of the path.  Thus, $G = \{(n-1) \times * | (n-1) \times *\} =  n \times *$.
  
  Now consider the case where $G$ is a path of $n$ uncolored vertices followed by a single blue vertex.  Again consider the moves for Left, starting with playing at the head and continuing down the tail: 
  \begin{align*}
    (n-2) \times * - 1, & \\
    (0 \times * - 1) + ((n-3) \times * - 1) & = (n-3) \times * - 2, \\
    (1 \times * - 1) + ((n-4) \times * - 1) & = (n-3) \times * - 2, \\
    \vdots, \\
    ((n-4) \times * - 1) + (1 \times * - 1) & = (n-3) \times * - 2, \\
    ((n-3) \times * - 1) + (0 \times * - 1) & = (n-3) \times * - 2
  \end{align*}
  All but the first simplify to $(n-3) \times * - 2$, worse for Left than the first option, $(n-2) \times * - 1$.
  
  Right has one more option: 
    \begin{align*}
      (n-2) \times * - 1,\\
      (0 \times * + 1) + ((n-3) \times * - 1) & = (n-3) \times *, \\
      (1 \times * + 1) + ((n-4) \times * - 1) & = (n-3) \times *, \\
      \vdots, \\
      ((n-3) \times * + 1) + (0 \times * - 1) & = (n-3) \times *, \\
      (n-2) \times * + 1
    \end{align*}
    All but the first and last simplify to $(n-3) \times *$.  $(n-2) \times * - 1 < (n-3) \times * < (n-2) \times * + 1$, so Right's best move is also to $(n-2) \times * - 1$.
  
  Now, 
  \begin{align*}
    G &= \{(n-2) \times * - 1 | (n-2) \times * - 1\} \\
    & = (n-2) \times * - 1 + *\\
    & = (n-1) \times * - 1
  \end{align*}
  The same steps can be used to show that a path ending in a red vertex has value $(n-1) \times * + 1$.
\end{proof}

\begin{corollary}
  On a path with $n$ uncolored vertices, it is always optimal to color the head.
\end{corollary}
\begin{proof}
  For all of our cases in our analysis above, the best move was always to color the vertex at the head.
\end{proof}

Now we can evaluate any path graph (or collection of paths).  The total value of the game is 

\begin{corollary}
\label{corol:pathMyopicCol}
    For a graph that consists only of paths, $G$, the total value of $M$, the \ruleset{Myopic Col} position on $G$ equals $a \times * + b - c$, where:
    \begin{itemize}
        \item $a$ is the number of uncolored vertices that either have no outgoing arc or have an outgoing arc that leads to another uncolored vertex.
        \item $b$ is the number of uncolored vertices with their outgoing arc leading to a red vertex.
        \item $c$ is the number of uncolored vertices with their outgoing arc leading to a blue vertex.
    \end{itemize}
\end{corollary}

\begin{proof}
    We can break each path into sections by removing edges leaving colored vertices.  By Lemma \ref{lem:pathMyopicColValues}, this splitting does not change the values of the positions at all.  Now the final formula, $a \times * + b - c$, is just the sum of all of the sums generated by using Lemma \ref{lem:pathMyopicColValues} on each piece.
\end{proof}

This formula is easy to evaluate, so we can solve this game efficiently.

\begin{corollary}[\ruleset{Myopic Col} on paths is in \cclass{P}]
    The winnability of \ruleset{Myopic Col} when played on path graphs can be determined in \cclass{P}.
\end{corollary}

\begin{proof}
    The formula from Corollary \ref{corol:pathMyopicCol} can be evaluated in $O(n)$ time, where $n$ is the number of vertices in the graph.
\end{proof}

In the next section, we put \ruleset{Myopic Col} to work again, but on slightly more complex graphs.  The computational complexity there is not yet known.

\section{Binary Trees}
\label{section:binary-trees}

In the same way that a linked list is a directed path graph, a binary tree is a directed tree graph with at most two outgoing arcs per vertex.  Each node in the tree (each vertex) contains a value in the data structure, as well as up to two references to other nodes.

Binary trees are used to implement quickly-searchable structures on ordered data, as well as heaps, which can be used for fast priority queues.

\subsection{Ruleset: Myopic Col on Trees}

We can reuse the ruleset for \ruleset{Myopic Col} from Section \ref{sec:myopiccol}, and now consider the game played on binary trees.  Unlike on paths, \ruleset{Myopic Col} on binary trees has non integer/integer-plus-star values.  Other dyadic rationals exist, such as $1/2$ and $-3/4$.

It seems natural that in a tree with an uncolored node as the root with two uncolored nodes as children of the root, coloring the root is the best move for both players, just as for a path that begins with two uncolored vertices.

\begin{conjecture}
    Consider a directed tree $T_0$ where the root vertex, $r$, is uncolored, and the child trees of $r$, $T_1$ and $T_2$, are both non-null and have uncolored roots.  For any tree $T$, let $G(T)$ be the \ruleset{Myopic Col} position on $T$.  Then $G(T_0) = * + G(T_1) + G(T_2)$.
\end{conjecture}






\section{Graphs}
\label{section:graphs}

Graphs covered in a data structures course are exactly the same as mathematical graphs with vertices, edges (or directed arcs), and values embedded in the vertices and/or edges.  There are a wide host of rulesets on graphs (\ruleset{Col}\cite{ONAG:2001}, \ruleset{Clobber}, \ruleset{Hackenbush}\cite{WinningWays:2001}, \ruleset{NimG}\cite{DBLP:journals/tcs/Fukuyama03}, and more).  Some games, such as \ruleset{Undirected Vertex Geography}, are in \cclass{P}\cite{DBLP:journals/tcs/FraenkelSU93}; others, such as \ruleset{Snort}, are \cclass{PSPACE}-complete\cite{DBLP:journals/jcss/Schaefer78}.

\section{Conclusions}

This paper presents game rulesets that rely heavily on different common data structures.  When relevant rulesets did not exist, new rulesets have been created: \ruleset{Tower Nim}, and \ruleset{Myopic Col}.  

We show polynomial-time algorithms that solve both \ruleset{Tower Nim} and \ruleset{Myopic Col} on paths.

\section{Future Work}

Three open problems exist concerning the computational difficulty of games presented here: \ruleset{Rotisserie Nim}, \ruleset{Anotnim}, and \ruleset{Myopic Col} on binary trees.

Although we have lots of results about \ruleset{Rotisserie Nim} positions, we don't yet have an efficient algorithm for the general state.

\begin{openProblem}[Rotisserie Nim]
    What is the computational complexity of \ruleset{Rotisserie Nim}?
\end{openProblem}

\ruleset{Antonim} is a classic game that has resisted a solution.  Recent work reduces the problem to a dynamic programming form\cite{arXiv:1506.01042v1}.  Unfortunately, since integers can be represented with a logarithmic number of bits for even a small number of heaps, the size of the table can be exponential in the number of bits needed to describe the game. Therefore, it remains unknown whether a polynomial-time algorithm to determine the outcome class of an \ruleset{Antonim} position exists.

\begin{openProblem}[Antonim]
    What is the computational complexity of \ruleset{Antonim}?
\end{openProblem}

For \ruleset{Myopic Col}, a polynomial-time solution exists on paths, but this does not immediately yield a solution on binary trees.

\begin{openProblem}[Binary Tree Myopic Col]
    What is the computational complexity of \ruleset{Myopic Col} played on binary trees?
\end{openProblem}

\bibliographystyle{plain}

\end{document}